\documentclass[12pt]{amsart}
\usepackage{amssymb,latexsym,amsmath,amscd,mathrsfs,yfonts,hyperref,euscript,upgreek}
\usepackage{fullpage}
\input xy
\xyoption{all}

\newcommand{\RR}{\mathbb R}
\newcommand{\ZZ}{\mathbb Z}

\newcommand{\CC}{\mathbb C}

\def\cC{\mathcal C}

\def\cD{\mathcal D}

\def\cQ{\mathcal Q}

\def\cF{\mathcal F}

\newcommand{\map}{\rightarrow}
\newcommand{\functor}{\longrightarrow}

\def\eC{\EuScript C}
\def\eD{\EuScript D}
\def\eF{\EuScript F}

\def\eK{\EuScript K}

\def\C{\textup{C}}

\def\idob{{\bf 1}}

\def\I{\textup{I}}
\def\id{\mathrm{id}}

\def\KKcat{\mathtt{KK}}
\def\Hom{\textup{Hom}}
\def\ker{\textup{ker}}

\def\fusion{\circledast}

\def\K{\textup{rk}}

\def\cob{\textup{Cob}}
\def\Cob{\textup{hCob}}
\def\Mod{\rm Mod}

\def\H{\textup{H}}

\def\K{\textup{K}}
\def\fk{\mathfrak K}
\def\KK{\textup{KK}}

\def\CT{\textup{CT}}
\def\Br{\textup{Br}}

\def\lim{{\varprojlim}}
\def\spinc{\textup{spin${}^c$}}

\newcommand{\beq}{\begin{eqnarray}}
\newcommand{\beqn}{\begin{eqnarray*}}
\newcommand{\eeq}{\end{eqnarray}}
\newcommand{\eeqn}{\end{eqnarray*}}

\newtheorem{thm}{Theorem}[section]

\newtheorem{lem}[thm]{Lemma}
\newtheorem{prop}[thm]{Proposition}
\newtheorem{cor}[thm]{Corollary}
\newtheorem{ex}{Example}
\newtheorem{defn}[thm]{Definition}
\newtheorem{rem}[thm]{Remark}

\newcommand{\comment}[1]{\noindent\fbox{\vtop{\noindent\textsf{\itshape #1}}}\\ \medskip \\}

\begin{document}

\title{Verlinde  modules and quantization}

\author{Varghese Mathai}
\address[V Mathai]{Department of Pure Mathematics, University of Adelaide, Adelaide,
SA 5005, Australia}
\email{mathai.varghese@adelaide.edu.au}

\begin{abstract}
Given a  compact simple Lie group $G$ and a primitive degree 3 twist $\eta$, 
we define a monoidal category $\eC(G, \eta)$  with a May structure given by disjoint union and fusion product. 
An object in the category $\eC(G, \eta)$ is a pair $(X, f)$, where $X$ 
is a compact $G$-manifold and $f: X\to G$ a smooth $G$-map with
respect to the conjugation action of $G$ on itself. Such an object
determines a
module, the twisted equivariant K-homology group $\K^G(X, f^*(\eta))$, for the Verlinde algebra, termed a Verlinde module, where the 
module action is induced by the $G$-action on $X$. 
In order to understand which objects in $\eC(G, \eta)$ can be quantized,
we define the closely related monoidal category $\eD(G, \eta)$ consisting of equivariant 
 twisted geometric K-cycles, 
which also has a May structure  given by disjoint union and fusion product. There is a forgetful functor $\eD(G, \eta)\map \eC(G, \eta)$,
showing that an object in $\eD(G, \eta)$ determines a Verlinde module.
Every object in the category $\eD(G, \eta)$ also
has a quantization, valued in the Verlinde algebra.
Finally, the quantization functor induces an isomorphism between the
geometric equivariant twisted K-homology ring 
$\K^G_{geo}(G, \eta)$, and the Verlinde algebra.
\end{abstract}

\thanks{{\em Acknowledgements.}
This research was supported under Australian Research Council's Discovery
Projects funding scheme (project number DP0878184).
V.M.\ is the recipient of an Australian Research Council Australian
Professorial Fellowship (project number DP0770927). V.M. thanks P. Hekmati and S. Mahanta for 
interesting discussions.}


\subjclass[2010]{}

\maketitle

%%%%%%%%%%%%%%%%%%%%%%%%%%%%%%%%%%%%%%%%%%%%%%%%%%%%%%

\section*{Introduction}
This work is partly inspired by a theorem of Freed, Hopkins and Teleman \cite{FHT1,FHT2,FHT3}, which identifies the twisted $G$-equivariant $\K$-homology group $\K^G(G, \eta)$  of a compact Lie group $G$, with the Verlinde algebra $R_\ell(G)$ of $G$ 
at a level $\ell$ determined by the twist $\eta$. Firstly, they show that $\K^G(G, \eta)$ is an algebra, with product induced by the multiplication $m : G\times G \longrightarrow G$ on the group $G$, and that their isomorphism ``explains'' the combinatorially complicated fusion product in the Verlinde algebra. Another consequence of their theorem is that the Verlinde algebra $R_\ell(G)$ has the same functorial properties as $\K^G(G, \eta)$. Recall that the importance of the Verlinde algebra is that it encodes the selection rules for the operator product expansion in certain rational conformal field theories such as the WZW-model. That is, they encode the dimensions of spaces of conformal blocks of these rational conformal field theories, i.e. dimensions of certain spaces of generalized theta functions (cf. \cite{Ver}). These dimensions and their polynomial behaviour are of fundamental importance in conformal field theory.

Another source of inspiration is the recent work by Meinrenken \cite{Mein1,Mein2} on the relation of quasi-Hamiltonian manifolds to the work of Freed, Hopkins and Teleman. To every compact quasi-Hamiltonian manifold $(M, \omega, \Phi)$ with group-valued moment map $\Phi: M \to G$ (which satisfies 
$\Phi^*(\eta)=d\omega$), Meinrenken defines the quantization $Q(M)$ to be the element of the 
Verlinde algebra $\Phi_*([M]) \in \K^G(G, \eta) \cong R_\ell(G)$, where $[M]$ denotes the equivariant fundamental class of the compact $G$-manifold $M$, which is an element in $\K^G(M, {\rm Cliff}(TM))$, since he shows that $M$ has an equivariant twisted Spinc structure, (explained later in the section). He then establishes several very interesting properties of his quantization procedure, as well as calculations of it. An approach related  to Meinrenken's, but which uses equivariant bundle gerbe modules instead, is due to Carey-Wang \cite{CW}.

We begin by proving some very general results, phrased using equivariant (operator) K-homology, which is equivariant K-homology on the category of separable $G$-$C^*$-algebras. It reduces to equivariant topological K-homology on compact topological spaces. Equivariant operator K-homology is a particular case of Kasparov's equivariant bivariant K-theory (or $\KK^G$-theory) which has a cup-cap product, making it into a powerful tool. We will utilise this cup-cap product and its properties, to prove 
our first general result, which says that if $A$ is a separable $G$-$C^*$-algebra which is also a $G$-coalgebra, then the Kasparov equivariant $K$-homology group $\KK^G(A, \CC)$ (as defined in \cite{Kas, Black}), admits a ring structure induced by the comultiplication on $A$, see Proposition
\ref{K-alg}. As a corollary,
we see that for a compact $C^*$-quantum group $A$, the K-homology group $\KK(A, \CC)$ admits a ring structure induced by the comultiplication on $A$, see Corollary \ref{q-alg}. As another corollary, we obtain a new proof of a theorem in \cite{FHT1, FHT2, FHT3} showing that the equivariant
twisted K-homology $\K^G(G, \eta)$, where $G$ acts on itself by conjugation and $\eta$ is a primitive degree 3 twist on $G$, is has 
a ring structure induced by the multiplication map $m:G\times G \map G$ on $G$. The next general result says that if  $A_1$ is a separable $G$-$C^*$-algebra which is also a $G$-coalgebra, and $A_2$ is another separable $G$-$C^*$-algebra which is also a
$G$-comodule for $A_1$,
then the abelian group $\K^G(A_2) =\KK^G(A_2, \CC)$ is a module for the algebra
$\K^G(A_1)=\KK^G(A_1, \CC)$, see Proposition \ref{K-mod}. As a corollary (see 
Corollary \ref{q-mod}) we see that  if  $A_1$ is a compact
$C^*$-quantum group and $A_2$ is a separable $C^*$-algebra which is also a comodule for $A_1$,
then the abelian group $\KK(A_2, \CC)$ is a module for the algebra $\KK(A_1, \CC)$.

We next define a category $\eC(G, \eta)$ whose objects are pairs $(X, f)$, where $X$ is a compact $G$-manifold and $f: X \to G$ is a smooth $G$-equivariant map, with respect to the conjugation action of $G$ acting on $G$. Corollary \ref{K-mod2} of Proposition \ref{K-mod} described above, establishes that the group action $a: G \times X \longrightarrow X$ determines
the module $\K^G(X, f^*(\eta))$ for the Verlinde algebra $R_\ell(G)$. We term such a module a {\em Verlinde module}, explaining part of the title of the paper. 
Moreover, any morphism in the category $\eC(G, \eta)$ determines a morphism of 
Verlinde modules over the
Verlinde algebra. 
In fact, we verify that $\eC(G, \eta)$ is a strict monoidal category
that has a {\em May structure}, which essentially says that the operations of disjoint union and fusion product are compatible, and can be 
viewed as an analogue of an algebra structure for categories.
Here the fusion product $\fusion$ is defined as $$(X_1,  f_1) \fusion (X_2,  f_2) = (X_1 \times X_2,  m\circ (f_1 \times f_2)),$$ for objects $(X_j,  f_j), \, j=1, 2$ in $\eC(G, \eta)$.
Denote the category of Verlinde modules by ${\rm VMod}(R_\ell(G))$, which is a subcategory of the category of all modules  ${\rm Mod}(R_\ell(G))$ over the Verlinde algebra.  We establish several interesting properties of ${\rm VMod}(R_\ell(G))$. More formally, we show that the lax monoidal functor 
$${\eF} : \eC(G, \eta) \longrightarrow {\rm Mod}(R_\ell(G)),$$ defined 
by ${\eF} (X, f) = \K^G(X, f^*(\eta))$, is compatible with the May
structures on both categories. These can be viewed as the central results in the paper.

Next we want to understand which objects in $\eC(G, \eta)$ can be quantized. To achieve this goal, we
define a closely related category $\eD(G, \eta)$, which objects are triples $(X, E,  f)$ where $(X, f)$ is an
object in the category $\eC(G,\eta)$ and
$E$ is a $G$-equivariant (complex) vector bundle over $X$.
In addition, we assume that $X$ has an {\em equivariant twisted Spinc structure}, that is, there is a given equivariant isomorphism,
\beq
f^*(\eK_\eta) \cong {\rm Cliff}(TX)\otimes \eK,
\eeq
where $\eK_\eta$ is the algebra bundle of compact operators on $G$ determined by $\eta$, $ {\rm Cliff}(TX)$ denotes the Clifford algebra bundle associated to the tangent bundle of $X$, and $\eK$ denotes
the algebra of compact operators on a $G$-Hilbert space.
The objects of $\eD(G, \eta)$ are an equivariant analogue of twisted geometric K-cycles in \cite{Baum, Wang}.
The key observation made here is that in this special case, this category has a richer structure than usual,
given by the May structure defined below.

$(G, {\bf 1}, \id:G\map G)$ is a final object in the category $\eD(G, \eta)$, where $\bf 1$ is the trivial line bundle over $G$. 
The morphisms of $\eD(G, \eta)$ are explicitly described in the text. In particular, a compact quasi-Hamiltonian $G$-manifold $(M, \omega, \Phi)$ determines the object $(M, {\bf 1}, \Phi)$ in $\eD(G, \eta)$, by a result in \cite{AM}. Clearly $\eD(G, \eta)$ is much larger, and it 
is closed under disjoint union $\coprod$, a dual operation, the fusion product $\fusion$ and 
also $G$-vector bundle modification, all of which will be explained in the text. Here we mention 
that the fusion product is $$(X_1, E_1, f_1) \fusion (X_2, E_2, f_2) = (X_1 \times X_2, E_1\boxtimes E_2, m\circ (f_1 \times f_2)),$$ for objects $(X_j, E_j, f_j), \, j=1, 2$ in $\eD(G, \eta)$.
We verify that $\eD(G, \eta)$ is a strict monoidal category
that has a {May structure} given by $\coprod$ and $\fusion$.

Every object $(X, E, f)$ in $\eD(G, \eta)$ has a {\em fundamental class} $[X]\in \K^G(X, {\rm Cliff}(TX))$, 
by a construction of Kasparov \cite{Kas}. The
{\em quantization} $\cQ(X, E, f) \in R_\ell(G)$ is defined as $\cQ(X, E, f) = f_*([E] \cap [X])$, which in the special case 
when $X$ itself is an equivariant Spinc manifold (i.e. $W_3^G(X) = W_3(X_G)=0$), 
reduces to $\cQ(X, E, f)= [f_*(\eth_X\otimes E)]$, where $\eth_X$ is the equivariant Spinc Dirac operator on $X$, and
$\eth_X\otimes E$ denotes the coupled operator. We will show that the quantization then determines a monoidal  functor, 
$$\cQ: \eD(G, \eta) \longrightarrow \K^G(G, \eta) \cong R_\ell(G),$$ 
respecting the May structures on both categories. In particular, an object $(X,  f)$ in $\eC(G, \eta)$ has a quantization if $X$ has an equivariant twisted Spinc structure, in which case $(X, {\bf 1}, f)$ defines an object in $\eD(G,\eta)$.

Consider the equivalence relation $\sim$ on objects in $\eD(G, \eta)$ generated isomorphism and by 
the following three elementary moves explained in the text,
\begin{enumerate}
\item direct sum-disjoint union;
\item $G$-bordism over $(G, \eta)$;
\item $G$-vector bundle modification.
\end{enumerate}
The geometric equivariant twisted K-homology group $\K^G_{geo}(G, \eta) $ is defined to be the abelian group which 
is the quotient $\eD(G, \eta)/\sim$, with addition induced by disjoint union-direct sum. Our first observation is that $\K^G_{geo}(G, \eta) $ has a ring structure induced by the fusion product $\fusion$ on $\eD(G, \eta)$.

The map induced by quantization is an isomorphism of rings,
\beq
\cQ :\K^G_{geo}(G, \eta)  \stackrel{\cong}\longrightarrow \K^G(G, \eta)\cong R_\ell(G).
\eeq
This is a special case of a more general theorem which will be proved elsewhere, which uses a hybrid of 
techniques in  \cite{Baum, Wang} and \cite{Baum2}. 

An impact of this 
result is that the May structure on the category $\eD(G, \eta)$ induces the algebra structure on $\K^G(G, \eta)$,
which by \cite{FHT1, FHT2, FHT3} is just the Verlinde algebra $R_\ell(G)$.

\section{The category $\eC(G, \eta)$}

Let $G$ be a connected, simply connected and simple compact Lie group with the multiplication map $m:G\times G\map G$. Let $G$ act on $G$ by the conjugation. Then it is known that every cohomology class $[\eta]\in\H^3_G(G, \ZZ)$ has a primitive representative, i.e., there is $\eta\in\Omega^3_\ZZ(G)$ a closed $3$-form on $G$ with integral periods, such that $m^*(\eta) = p_1^*(\eta) + p_2^*(\eta)$. 

 Let $\eC(G,\eta)$ denote the category, whose objects are pairs 
$(X, f)$, where $X$ is a compact $G$-manifold and $f:X\map G$ is a $G$-equivariant map
with respect to the conjugation action of $G$ acting on $G$.

A morphism $\theta$ between the objects $(X_1, f_1)$ and $(X_2, f_2)$ 
in the category $\eC(G, \eta)$ is a (pointed) smooth $G$-map 
$\theta:X_1\map X_2$ compatible with the structure maps, i.e., the following diagram commutes
\beq
\xymatrix{
X_1 \ar[ddr]_{f_1} \ar[rr]^{\theta} && X_2
\ar[ddl]^{f_2} \\ &&\\
& G   }  
\eeq
In particular, the objects $(X_1, f_1)$ and $(X_2, f_2)$ 
in the category $\eC(G, \eta)$ are said to be {\em isomorphic} if there is an isomorphism $\theta:  
(X_1, f_1) \map(X_2, f_2)$ in $\eC(G, \eta)$.
Observe that $(G,\id:G\map G)$ is a final object in this category.

The chosen $\eta$ determines a canonical class in $\H^3(X,\ZZ)$ for any object $(X,f)$ given by the pullback along the structure map $f^*(\eta)$. We list below several interesting objects of the category 
$\eC(G,\eta)$. More examples will be given in the appendix.

\begin{ex} (Trivial action) 

Let $X$ be a $G$-space with the trivial $G$-action, 
and $f: X \to G$ be a continuous map
such that the image of $f$ lies in the centre of $G$. Then $f$ is equivariant with respect to the
adjoint action of $G$ on $G$, so that $(X,  f)  \in \eC(G,\eta)$ for any $G$-vector bundle $E$ over $X$.

\end{ex}

\begin{ex} (Free action)  

Let $X$ be a free $G$-space, 
and $f: X/G \to G/G$ be a continuous map to the space of conjugacy classes $G/G$ in $G$.
Then $f$ lifts to an equivariant map $\tilde f\colon X \to G$ with respect to the
adjoint action of $G$ on $G$, so that $(X,  f)  \in \eC(G,\eta)$.

\end{ex}

\begin{ex}(Hamiltonian $G$-spaces)

Let $(X, \omega)$ be a Hamiltonian $G$-space with moment map $\mu: X \longrightarrow 
{\mathfrak  g}^*$. If $G$ has an Ad-invariant metric, then ${\mathfrak  g}^*\cong {\mathfrak  g}$,
which will be assumed here.

Then $(X, f) \in \eC(G, \eta)$, where $f = \exp(\mu)$, where $\exp: {\mathfrak  g} \longrightarrow G$
denotes the exponential map of the Lie group $G$.

\end{ex}
 
\begin{ex} (Inverse)

Let $(X,  f) \in \eC(G, \eta)$. Then  $(X,  1/f) \in \eC(G, \eta)$.

\end{ex}

\begin{ex}(Fusion product) \label{ex:fusion}

Let $(X_1,  f_1), (X_2,  f_2)$ be objects in $\eC(G, \eta)$. Then  $(X_1 \times X_2,  f_1 \times f_2) \in  
\eC(G \times G, p_1^*(\eta) + p_2^*(\eta) )$, and $(X_1 \times X_2, m \circ (f_1 \times f_2)) \in \eC(G, \eta)$
is called the {\em fusion of $(X_1, f_1)$ and $(X_2, f_2)$} and is denoted by $(X_1, f_1) \circledast (X_2, f_2) $.
Here we use the fact that $m^*(\eta) =  p_1^*(\eta) + p_2^*(\eta) $.

\end{ex}

\begin{ex}(Disjoint union) \label{ex:disjoint}

Let $(X_1,  f_1), (X_2,  f_2)$ be objects in $\eC(G, \eta)$. The disjoint union
$(X_1 \coprod X_2,  f_1 \coprod f_2) \in  
\eC(G, \eta)$.

\end{ex}

\subsection{Morphisms in the category $\eC(G, \eta)$ and more examples}

A morphism $\phi : (X_1, f_1) \longrightarrow  (X_2, f_2)$, where $(X_i, f_i) \in \eC(G_i, \eta_i), \, i=1,2$, is given by a group homomorphism $\phi_G : G_1 \longrightarrow G_2$ such that $\phi_G^*(\eta_2) = \eta_1$ and a smooth map $\phi_X : X_1 \longrightarrow X_2$ satisfying
$$
f_2 \circ \phi_X = \phi_G\circ f_1, \qquad \phi_X(g.x) = \phi_G(g).\phi_X(x)
$$
for all $g\in G_1$ and $x\in X_1$. In particular, when $G_1=G_2=G$  and $\eta_1=\eta_2=\eta$, we obtain the morphisms in the category $\eC(G, \eta)$. 

\begin{ex}(Fixed point set)

Suppose that a compact Lie group $K$ acts by automorphisms on $(X,  f) \in \eC(G, \eta)$. Then the fixed-point set 
$X^K$ and the restriction of $f$, $f^K: X^K \to G^K$ determines an object 
$(X^K,  f^K) \in \eC(G^K,  \eta^K)$ and is 
such that the inclusion $\iota : X^K \hookrightarrow X$ is a morphism.

\end{ex}

\begin{ex} (Fixed point set - special case)

As a special case of the example above, suppose that 
$(X,  f)$ is an object in $\eC(G, \eta)$ and $g \in G$. Then the components of the fixed point set 
of $g$, $(X^g,  f^g) \in \eC(G^g, \eta^g)$ is 
such that the inclusion $\iota : X^g \hookrightarrow X$ is a morphism.

\end{ex}

\begin{ex} (Finite covers)

Suppose that 
$(X,f)$ is an object in $\eC(G, \eta)$ and $p: \widetilde G \to G$ a finite cover of $G$.
Consider the fibre product $\widetilde X = \{ (x, \tilde g) \in X \times  \widetilde G\, \big| \, f(x)=p(\tilde g)\}$
and $\widetilde f = p_1^*f$, where $p_1:\widetilde X \to X$ is the projection onto the first factor.
Then $(\widetilde X, \widetilde f)$ is an object in  $\eC(G, \eta)$ such that the quotient map $\widetilde X\mapsto X$
is a morphism in $\eC(G, \eta)$. 
\end{ex}

\begin{ex} (Quasi-Hamiltonian $G$-spaces and Hamiltonian $LG$-spaces)

 We recall the definition from \cite{AMM} of  the
   definition of a group-valued moment map for a
   quasi-Hamiltonian $G$-space.
A {\bf quasi-Hamiltonian $G$-manifold} is a $G$-manifold $M$ with
     an invariant 2-form $\omega \in \Omega^2(M)^G$ and an equivariant
     map $\Phi \in C^\infty(M, G)^G$ such that
     \begin{enumerate}
     \item $d\omega = \Phi^*(\eta)$;
     \item The map $\Phi$ satisfies $\iota(v_\xi) = \displaystyle{\frac k2} \Phi^*\langle
     \theta + \bar \theta, \xi \rangle,$ where $v_\xi$ is the vector
     field on $M$ generated by $\xi\in {\frak g}$, $\theta$ is the left invariant 
     Cartan-Maurer form, $\bar\theta$ is the right invariant Cartan-Maurer form;
     \item At each $x\in M$, the kernel of $\omega_x$ is given by
      \[
      \ker (\omega_x) = \{ v_\xi| \ \xi \in ker(Ad_{\Phi(x)+1})\}.
      \]
    The map $\Phi$ is called the group valued moment map of the
    quasi-Hamiltonian $G$-manifold $M$. It is proved  in \cite{AM} that 
    $(M,  \Phi)$ is an element in $\eC(G, \eta)$.
    \end{enumerate}

    Basic examples of quasi-Hamiltonian $G$-spaces are provided by 
    products of conjugacy classes
    ${\mathcal C }\subset G$ as in \cite{AMM}. More generally, a
    bijective correspondence between Hamiltonian loop group manifolds
    with proper moment map and quasi-Hamiltonian $G$-manifold is established
    in that paper. They constitute a large number of objects in $\eC(G, \eta)$.

\end{ex}

\begin{ex}(Disjoint union)

The disjoint union of a pair of objects $(X_i,  f_i), \, i=1,2$ in $\eC(G, \eta)$ is defined 
as the disjoint union,
$$
(X_1,  f_1) \oplus (X_2,  f_2) = (X_1 \coprod X_2,  f_1\coprod f_2)
$$
\end{ex}

\begin{ex}(Extension by maps)

Let $(X,  f)$ be an object in $\eC(G, \eta)$ and $h: Y\longrightarrow X$ be a $G$-map. 
Then $(Y, f\circ h) \in \eC(G, \eta)$.

For example, let $W$ be an equivariant vector bundle over $X$. 
Then $(\widehat X, f\circ \pi)$ is again an object in $\eC(G, \eta)$. 
Here $\widehat X$ denotes the unit sphere
bundle $S((X\times \RR)\oplus W)$, $\pi: \widehat X \longrightarrow X$ is the projection.

\end{ex}

\subsection{The fusion monoidal structure on $\eC(G, \eta)$}
Recall that a monoidal category is a category in which associated to each 
pair of objects $A$ and $B$ there exists 
a product object $A\otimes B$, and there is an identity object $\idob$, 
such that $\idob\otimes A \cong A \cong A\otimes \idob$,
together with associator isomorphisms
$$
\Phi = \Phi_{A,B,C}: A\otimes (B\otimes C) \to (A\otimes B)\otimes C
$$
for any three objects $A$, $B$ and $C$, satisfying the    {\em  pentagonal identity}:
{\small
\begin{equation*}
\xymatrix@=1pc@C-40pt{
&& A\otimes (B\otimes(C\otimes D)) \ar[ddll]_{\idob_A\otimes \Phi_{B,C,D}}  
\ar[ddrr]^{\Phi_{A,B,C\otimes D}} && \\ &&&& \\
A\otimes ((B\otimes C)\otimes D) \ar[dddr]^{\Phi_{A,B\otimes C,D}} && &&
(A\otimes B)\otimes (C\otimes D) \ar[dddl]_{\Phi_{A\otimes B,C,D}}  \\ &&&& \\ &&&& \\
& (A\otimes (B\otimes C))\otimes D \ar[rr]_{\Phi_{A,B,C} \otimes \idob_D} && 
((A\otimes B)\otimes C)\otimes D & }
\end{equation*}}
with each arrow the appropriate map $\Phi$, and the triangle relation:
\begin{equation*}
\xymatrix{
A\otimes(\idob\otimes B) \ar[ddr]_{\cong} \ar[rr]_{\Phi_{A,\idob,B}} && (A\otimes \idob)\otimes B 
\ar[ddl]^{\cong} \\ &&\\
& A\otimes B   }  
\end{equation*}
{\em MacLane's coherence theorem} ensures that these conditions are sufficient to 
guarantee consistency of all other rebracketings.

Recall the fusion product in $\eC(G,\eta)$ from Example \ref{ex:fusion}. Given $(X_1,  f_1)$ and $(X_2, f_2)$ the following diagram describes the fusion product $$(X_1, f_1)\fusion (X_2, f_2)=(X_1\times X_2, m\circ(f_1\times f_2)).$$

\beqn
\xymatrix{
X_1\times X_2 \ar@/_/[ddr]_{p_{X_1}} \ar@/^/[drr]^{p_{X_2}}
\ar@{->}[dr]|-{f_1\times f_2} \\
& G \times G \ar[dr]^m
& X_2 \ar[d]^{f_2} \\
& X_1 \ar[r]_{f_1} & G 
}
\eeqn 

\begin{prop}
 $\eC(G, \eta)$ is a strict monoidal category with product given by the fusion product $\circledast$ (as in Example \ref{ex:fusion}) and identity element $(e, \I)$, where $e\in G$ is the identity element and $\I: e\hookrightarrow G$ is the inclusion map,
 which is equivariant under the adjoint action of $G$ on itself and the trivial action of $G$ on $e$.
\end{prop}

\begin{proof} We need to verify that  $\eC(G, \eta)$ satisfies the axioms of a monoidal category. 

Let $(X,  f) \in \eC(G, \eta)$ and consider the fusion product  $$(X, f) \circledast (e,  \I)
=  (X \times e,  m \circ (f \times \I))$$ which is clearly isomorphic to $(X, f)$. Similarly, 
$(e, \I) \circledast (X,  f)$ is also isomorphic to $(X, f)$. Therefore  $(e,  \I) \in \eC(G, \eta)$
serves as an identity in $\eC(G, \eta)$.

Next, let $(X_i, f_i) \in \eC(G, \eta), \, i=1,2,3$, and 
$$
\Phi_{1,2,3}: (X_1, f_1) \fusion ((X_2, f_2) \fusion  (X_3, f_3) ) \to 
((X_1,  f_1) \fusion (X_2,  f_2)) \circledast  (X_3,  f_3).
$$
The LHS  is canonically isomorphic to 
$$
(X_1 \times X_2 \times X_3, m\circ (f_1 \times m\circ (f_2 \times f_3))),
$$
whereas the RHS is canonically isomorphic to 
$$
(X_1 \times X_2 \times X_3, m\circ(m \circ (f_1 \times f_2) \times f_3)).
$$
The equality $m\circ (f_1 \times m\circ (f_2 \times f_3)) = m\circ(m \circ (f_1 \times f_2) \times f_3)$ 
follows from the associativity of the multiplication $m$ on $G$. In particular, $\Phi_{1,2,3}=\id$
and $\eC(G, \eta)$ is a strict category. In particular, the pentagonal identity is automatically 
satisfied, and by  MacLane's coherence theorem, $\eC(G, \eta)$ is a strict monoidal category.

\end{proof}

 \subsection{May structure}
Motivated by the study of infinite loop spaces and algebraic $\K$-theory J. P. May introduced the machinery of (bi)permutative categories \cite{May1,May2}. These are, roughly, categories equipped with two monoidal structures  satisfying certain compatibility conditions. We have seen before the fusion monoidal structure on $\eC(G,\eta)$. 

\begin{defn}
A May structure on a category $\cC$ is a pair of monoidal structures $\oplus$, $\otimes$ on $\cC$, such that for all $X,Y,Z\in\cC$ there is a natural (distributivity) isomorphism $X\otimes(Y\oplus Z)\cong (X\otimes Y) \oplus (X\otimes Z)$. We call a category equipped with a May structure as a May category.

A May functor $F:\cC\functor\cD$ between May categories is one which is lax monoidal with respect to both $\oplus$ and $\otimes$, such that, in addition, the following diagram commutes:

\beqn
\xymatrix{
F(X)\otimes F(Y) \oplus F(X)\otimes F(Z)\ar[r]\ar[d] & F(X\otimes Y)\oplus F(X\otimes Z) \ar[r] & F((X\otimes Y) \oplus (X\otimes Z))\ar[d] \\
F(X)\otimes(F(Y)\oplus F(Z))\ar[r] & F(X)\otimes F(Y\oplus Z)\ar[r] & F(X\otimes(Y\oplus Z))
}
\eeqn Here the vertical maps are the distributivity isomorphisms and the horizontal maps are furnished by the lax monoidal structures.
\end{defn}

\begin{ex}
The category $\eC(G,\eta)$ endowed with $\oplus=\coprod$ and $\otimes=\fusion$ is a May category. We have already seen that disjoint union and fusion product are monoidal structures on $\eC(G,\eta)$. The distributivity condition is readily verified.
\end{ex}

\section{Continuous trace $C^*$-algebras on $G$ and $G$-coalgebra structure}

Let $G$ be a compact Lie group and let $X$ be a compact $G$-space. Let $\eK$ denote the $C^*$-algebra of compact operators on a separable $G$-Hilbert space. A $G$-equivariant {\em Dixmier--Douady (DD) bundle} on $X$ is a $G$-equivariant algebra bundle on $X$, whose fibres are $\eK$ with the projective unitary group ${\rm PU}$ as the structure group. Given any $G$-equivariant principal ${\rm PU}$ bundle $P$ over $X$, one can construct an a DD-bundle over $X$ as an associated bundle 
$\eK_P = P\times_{{\rm PU}}\eK$. The equivariant Dixmier-Douady class of a $G$-equivariant { Dixmier--Douady (DD) bundle} on $X$ 
is a cohomology class in $\H^3_G(X, \ZZ)$. A recent theorem of Atiyah-Segal \cite{AS} says that every equivariant cohomology class in $\H^3_G(X, \ZZ)$ determines a $G$-equivariant ${\rm PU}$-bundle $P$ over $X$ up to 
$G$-equivariant isomorphism, which determines the 
$G$-equivariant DD-bundle $\eK_P$, thus establishing the converse. In fact, the $G$-equivariant isomorphism classes of 
DD-bundles $\eK_P$ over $X$ form an abelian group under direct sum, called the equivariant Brauer group, ${\Br}_G(X)$. Since the equivariant Dixmer-Douady class of a direct sum of $G$-equivariant DD-bundles is equal to the sum of the  equivariant Dixmer-Douady classes of the $G$-equivariant DD-bundles, there is a natural isomorphism of groups, ${\Br}_G(X)\cong \H^3_G(X, \ZZ)$.

Consider a category whose objects are pairs $(X,P)$, where $X$ is a compact $G$-space and $P$ is a fixed choice of a $G$-equivariant principal ${\rm PU}$-bundle on $X$. A morphism $(X,P)\map (X',P')$ in this category is a continuous $G$-map $f:X\map X'$, such that $f^*(P')= P$. 
Let $G,X$ be as above and let $P$ be a $G$-equivariant principal ${\rm PU}$-bundle on $X$. Let $\CT(X,P)$ denote the stable continuous trace $C^*$-algebra of all continuous sections of the associated equivariant DD bundle $\eK_P$ over $X$, that vanish at infinity. The induced  $G$-action on $\CT(X,P)$ makes it into a $G$-$C^*$-algebra, which enables us to construct the crossed product $C^*$-algebras $\CT(X,P)\rtimes G$. The association $(X,P)\mapsto \CT(X,P)\rtimes G$ is functorial with respect to the above-mentioned morphisms of pairs. Furthermore, if $(X,P)$ and $(X,P')$ are two pairs, such that the Dixmier--Douady invariants of $P$ and $P'$ determine the same class in $\H^3_G(X,\ZZ)$, then the equivariant Dixmier--Douady Theorem says that the $C^*$-algebras $\CT(X,P)\rtimes G$ and $\CT(X,P')\rtimes G$ are stably isomorphic \cite{AS}. 

Given a compact simple Lie group $G$ with multiplication map $m:G\times G\map G$, a $G$-equivariant DD-bundle $\eK_P$ on $G$ is said to be {\em primitive} if $m^*(\eK_P) \cong p_1^*(\eK_P)\otimes p_2^*(\eK_P)$ as $G$-equivariant bundles on $G\times G$. The equivariant Dixmier-Douady class $\eta $ of a primitive $G$-equivariant DD-bundle $\eK_P$ satisfies the corresponding primitivity property, $m^*(\eta)= p_1^*(\eta)+p_2^*(\eta)$. There are many natural ways to construct primitive $G$-equivariant DD-bundle $\eK_P$ on $G$. One of these uses positive energy representations \cite{AS, FHT1, FHT2, FHT3} and another more constructive method is in \cite{Mein2}. In this case, we denote the equivariant continuous trace algebra by $\CT(G, \eta)$. If $(X, f)$ is an object in the category $\eC(G, \eta)$, then we denote the corresponding  equivariant continuous trace algebra by $\CT(X, f^*(\eta))$.

We next prove some fundamental facts about the equivariant continuous trace algebras $\CT(G, \eta)$ and $\CT(X, f^*(\eta))$.

\begin{lem}\label{lem:coalgstructonCT}
Let $G$ be a compact group, and $\eta$ be a primitive DD-bundle on $G$. Then the $C^*$-algebras $\CT(G, \eta)$ and $\CT(G, \eta)\rtimes G$ both carry a natural coalgebra structure induced by the multiplication on $G$.
\end{lem}

\begin{proof}
The multiplication map $m:G\times G\map G$ induces a $*$-homomorphism 
$$m^*:\CT(G, \eta)\map \CT(G\times G, m^*(\eta)) \cong \CT(G\times G, p_1^*(\eta) + p_2^*(\eta))\cong\CT(G, \eta)\otimes\CT(G, \eta).$$
Since the multiplication map $m$ and the $C^*$-tensor product $\otimes$ are associative, it follows that the induced 
$*$-homomorphism $m^*$ is coassociative, that is, the diagram below commutes,
\begin{equation*} \label{cup-cap-associativity}
\xymatrix @=8pc @ur { \CT(G, \eta)\otimes \CT(G, \eta) \ar[d]_{1\otimes m^*} & 
\CT(G, \eta) \ar[d]_{m^*} 
\ar[l]^{m^*} \\  \CT(G, \eta)\otimes \CT(G, \eta)\otimes \CT(G, \eta) & 
\CT(G, \eta)\otimes \CT(G, \eta)
\ar[l]^{m^* \otimes 1}}
\end{equation*}

Moreover, the multiplication map $m\colon G\times G\map G$ is $G$-equivariant under the adjoint action of $G$. Therefore, it induces a $*$-homomorphism of $G$-$C^*$-algebras and in turn that of the crossed product algebras, $$m^*:\CT(G, \eta)\rtimes G\map (\CT(G, \eta)\rtimes G)\otimes(\CT(G, \eta)\rtimes G).$$ Coassociativity of $m^*$ follows by similar reasoning as above.
\end{proof} 

\noindent
Using analogous arguments one easily proves the following result:

\begin{lem} \label{comod}
Let $(X, f)$ be an object in the category $\eC(G, \eta)$. Then $\CT(X,f^*(\eta))$ (resp. $\CT(X,f^*(\eta))\rtimes G$) is a comodule over the coalgebra $\CT(G,\eta)$ (resp. $\CT(G,\eta)\rtimes G$), which is induced by the group action map.
\end{lem}

\begin{proof}
The group action map $a:G\times X\map X$ is a $G$-map under the adjoint action of $G$, therefore it induces a $*$-homomorphism  of $G$-$C^*$-algebras,
 $$
 a^*:\CT(X, f^*(\eta))\map \CT(G \times X, a^*f^*(\eta)) \cong  \CT(G \times X, p_1^*\eta + p_2^*f^*(\eta))
\cong \CT(G, \eta)\otimes \CT(X, f^*(\eta))
 $$
 since $a^*(f^*\eta) = p_1^*(\eta) + p_2^*f^*(\eta)$ by the primitivity assumption on $\eta$.
 The defining property of the action map $a$ shows that $a^*$ is a comodule map, that is, the diagram below commutes,
\begin{equation*} \label{cup-cap-associativity}
\xymatrix @=8pc @ur { \CT(G, \eta)\otimes \CT(X, f^*(\eta)) \ar[d]_{1\otimes a^*} & 
\CT(X, f^*(\eta))\ar[d]_{a^*} 
\ar[l]^{a^*} \\  \CT(G, \eta)\otimes \CT(G, \eta)\otimes \CT(X, f^*(\eta)) & 
\CT(G, \eta)\otimes \CT(X, f^*(\eta))
\ar[l]^{a^* \otimes 1}}
\end{equation*}

Moreover, the group action map $a\colon G\times X\map X$ is $G$-equivariant under the adjoint action of $G$. Therefore, it induces a $*$-homomorphism of $G$-$C^*$-algebras and in turn that of the crossed product algebras, 
$$a^*:\CT(X, f^*(\eta))\rtimes G\map (\CT(G, \eta)\rtimes G)\otimes(\CT(X, f^*(\eta))\rtimes G).$$ Coassociativity of $a^*$ follows by similar reasoning as above.
\end{proof}

\section{The ring structure on $\K^G(G, \eta)$ and Verlinde modules}

We will derive the ring structure on $\K^G(G, \eta)$ as a special instance of a more 
general result. We will also show that for $(X, f) \in \eC(G, \eta)$, $\K^G(X, f^*(\eta))$
is a module for the ring $\K^G(G, \eta)$ as a special case of a more general result.
The functor $\eF : \eC(G, \eta) \longrightarrow {\rm Mod}(R_\ell(G))$ defined by 
$\eF(X,f)=\K^G(X, f^*(\eta))$ is shown to be a lax monoidal 
which respects the May structures on both categories.

The very general result that we will prove in this section this described as follows. Recall that if $A$ is a 
separable $G$-$C^*$-algebra, then
the equivariant K-homology $\K^G(A)$ is in general only an abelian group. However, if 
$A$ is also a $G$-coalgebra, then we will show that $\K^G(A)$ is a ring, with product induced
by the comultiplication on $A$. In particular, we deduce that the K-homology of a $C^*$-quantum group 
is an ring, which appears to be new. We also prove in an analogous way that the equivariant 
K-theory $\K_G(A)$ is a coring, with coproduct induced
by the comultiplication on $A$. In particular, we deduce that the K-theory of a $C^*$-quantum group 
is a coring, which also appears to be new. Finally, if $A$ is also K-oriented in equivariant K-theory, then 
both $\K^G(A)$ and $\K_G(A)$ are bialgebras in a natural way.

A key step in defining the product in $\K^G(A)$ uses Kasparov's cup-cap product. More precisely, let $A_1, A_2, D, B_1, B_2$ be separable $G$-$C^*$-algebras,
where $G$ is a locally compact group. 
Then the {\em cup-cap product}, (Definition 2.12 of \cite{Kas})
\beq\label{cup-cap}
\otimes_D : \KK^G(A_1, B_1\otimes D) \otimes  \KK^G(D\otimes A_2, B_2) \longrightarrow  \KK^G(A_1 \otimes A_2, B_1\otimes B_2)  
\eeq

\begin{prop}\label{K-alg}
Let $A$ be a separable $G$-$C^*$-algebra which is also a $G$-coalgebra. Then the 
 abelian group $\KK^G(A, \CC)$ admits a ring structure 
induced by the comultiplication on $A$. 
\end{prop}

\begin{proof}
When $B_1=\CC = B_2=D$  and $A_1=A_2=A$, the cup-cap product  \eqref{cup-cap}
reduces to 
\beq\label{cup-cap2}
\otimes_\CC : \KK^G(A, \CC) \otimes  \KK^G(A, \CC) \longrightarrow  \KK^G(A \otimes A, \CC).
\eeq
The comultiplication $\Delta$ on $A$ is a $G$-$*$-homomorphism
\beq
\Delta: A \longrightarrow A \otimes A.
\eeq
Since equivariant-K-homology is contravariant with respect to $G$-$*$-homomorphisms, we get
\beq\label{comult}
\Delta^*: \KK^G(A \otimes A, \CC) \longrightarrow  \KK^G(A, \CC).
\eeq
The composition of the morphisms in equations \eqref{cup-cap2} and  \eqref{comult} gives
$$
\circ: \KK^G( A, \CC) \otimes  \KK^G( A, \CC) \longrightarrow  \KK^G( A, \CC),
$$
which is a product $\circ$ on equivariant K-homology $ \KK^G( A, \CC) $ induced by the comultiplication. The associativity of the product $\circ$  follows from the associativity of the cup-cap product (Theorem 2.14 of \cite{Kas}),
the coassociativity of the comultiplication $\Delta$ on $A$ and the naturality of the cup-cap product under $G$-$C^*$-homomorphisms. 

More precisely, the associativity of Kasparov's cup-cap product says that the following diagram commutes,
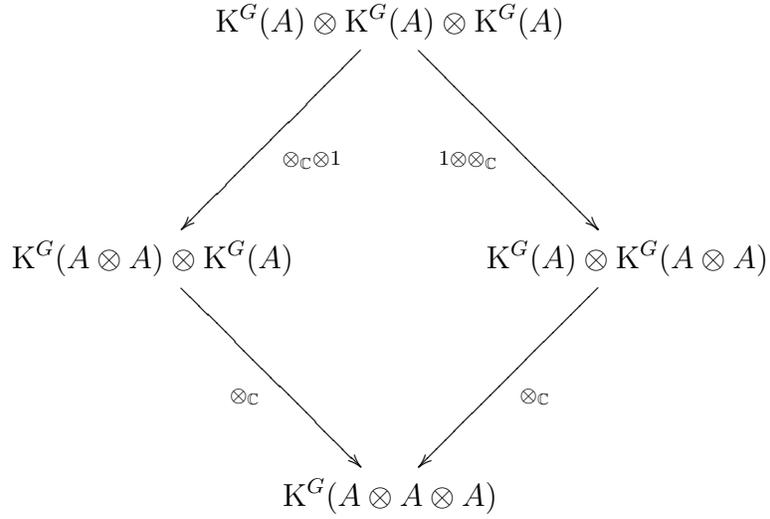
\begin{equation*} \label{cup-cap-associativity}
\xymatrix @=8pc @ur { \K^G(A\otimes A) \otimes \K^G(A) \ar[d]_{\otimes_\CC} & 
\K^G(A)\otimes \K^G(A)\otimes \K^G(A) \ar[d]_{1\otimes \otimes_\CC} 
\ar[l]^{\otimes_\CC \otimes 1} \\ \K^G(A\otimes A\otimes A) & 
\K^G(A)\otimes \K^G(A\otimes A)
\ar[l]^{\otimes_\CC}}
\end{equation*}
where $\K^G(A)$ denotes $\KK^G(A, \CC)$.

On the other hand, the coassociativity of the comultiplication $\Delta$ says that the following diagram commutes,
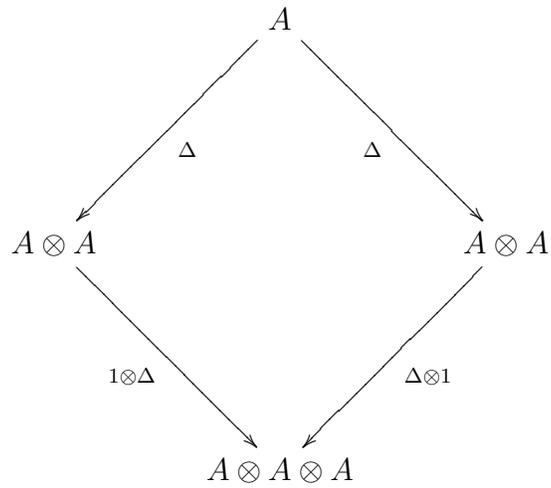
\begin{equation*} \label{cup-cap-associativity}
\xymatrix @=8pc @ur { A\otimes A \ar[d]_{1\otimes \Delta} & 
A \ar[d]_{\Delta} 
\ar[l]^{\Delta} \\  A\otimes A\otimes A & 
A\otimes A
\ar[l]^{\Delta \otimes 1}}
\end{equation*}

Therefore the induced diagram in equivariant K-homology commutes

\begin{equation*} \label{cup-cap-associativity}
\xymatrix @=8pc @ur { \K^G(A\otimes A) \ar[d]_{\Delta^*} & 
\K^G(A\otimes A\otimes A) \ar[d]_{1\otimes \Delta^*} 
\ar[l]^{\Delta^* \otimes 1} \\  \K^G(A) & 
\K^G(A\otimes A)
\ar[l]^{\Delta^*}}
\end{equation*}

Therefore one has the commutative diagram,

\beqn
\xymatrix{
\K^G(A)\otimes \K^G(A)\otimes \K^G(A) \ar[r]^{1\otimes(\otimes_\CC)}\ar[d]_{(\otimes_\CC)\otimes 1} & \K^G(A)\otimes \K^G(A\otimes A)\ar[r]^{1\otimes \Delta^*}\ar[d]_{\otimes_\CC} & \K^G(A)\otimes \K^G(A)\ar[d]_{\otimes_\CC} \\
\K^G(A\otimes A)\otimes \K^G(A) \ar[r]^{\otimes_\CC}\ar[d]_{\Delta^*\otimes 1} & \K^G(A\otimes A\otimes A) \ar[r]^{1\otimes \Delta^*}\ar[d]_{\Delta^*\otimes 1} & \K^G(A\otimes A)\ar[d]_{\Delta^*}\\
 \K^G(A)\otimes \K^G(A) \ar[r]^{\otimes_\CC} & \K^G(A\otimes A) \ar[r]^{\Delta^*} & \K^G(A)
}
\eeqn

The top left hand square commutes since the cup-cap product is associative, the bottom right hand square 
commutes since the comultiplication is coassociative, while the remaining squares commute because 
the cup-cap product is functorial under $G$-$C^*$-homomorphisms.

Therefore one has the commutative diagram,

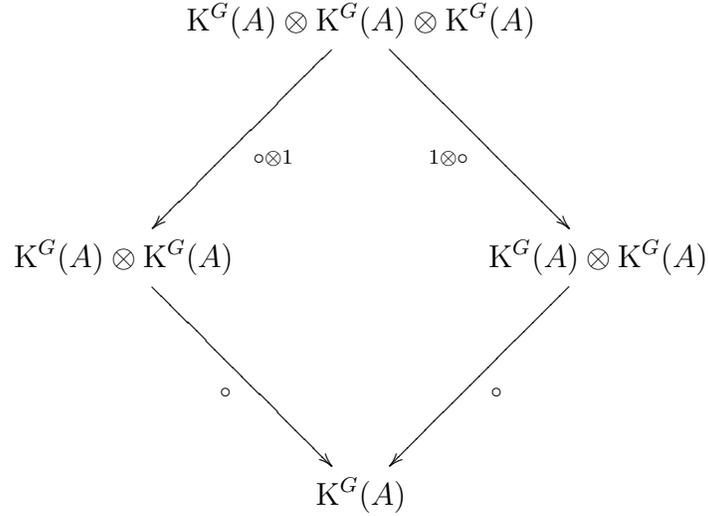
\begin{equation*} \label{circ-associativity}
\xymatrix @=8pc @ur { \K^G(A) \otimes \K^G(A) \ar[d]_{\circ} & 
\K^G(A)\otimes \K^G(A)\otimes \K^G(A) \ar[d]_{1\otimes \circ} 
\ar[l]^{\circ \otimes 1} \\ \K^G(A) & 
\K^G(A)\otimes \K^G(A)
\ar[l]^{\circ}}
\end{equation*}
which precisely says that product $\circ$ is associative.

\end{proof}

The following corollary was was first established in \cite{FHT1,FHT2,FHT3}. See also \cite{TuXu}.

\begin{cor}
Let $G$ be a compact Lie group, and consider the congugation action of $G$ on itself.
Let $[\eta]\in\H^3_G(G,\ZZ)$ a primitive cohomology class. 
Then the abelian group $\K^G(G,\eta) =\KK^G(\CT(G, \eta), \CC)$ admits a  ring structure 
induced by the multiplication $m : G \times G \longrightarrow G$ on $G$. 
\end{cor}

\begin{proof}
Setting $A = \CT(G, \eta)$ in Proposition \ref{K-alg}, we need to show that $A$ is a $G$-coalgebra.

The multiplication map $m\colon G\times G\map G$
 is equivariant under the adjoint action of $G$, therefore it induces a $*$-homomorphism  of $G$-$C^*$-algebras which is a comultiplication,
 $$
 m^*:\CT(G, \eta)\map \CT(G, \eta)\otimes \CT(G, \eta).
 $$
 Since the multiplication map $m$ and the $C^*$-tensor product $\otimes$ are associative, it follows that the induced 
$*$-homomorphism $m^*$ is coassociative.

\end{proof}

The following is an immediate corollary of Proposition \ref{K-alg}.

\begin{cor}\label{q-alg}
Let $A$ be a compact $C^*$-quantum group. Then the 
 abelian group $\KK(A, \CC)$ admits a ring structure 
induced by the comultiplication on $A$. 
\end{cor}

We next prove another general result.

\begin{prop}\label{K-mod}
Let $A_1$ be a separable $G$-$C^*$-algebra which is also a $G$-coalgebra. 
Let $A_2$ be another separable $G$-$C^*$-algebra which is also a
$G$-comodule for $A_1$.
Then the abelian group $\K^G(A_2) =\KK^G(A_2, \CC)$ is a module for the algebra
$\K^G(A_1)=\KK^G(A_1, \CC)$.
\end{prop}

\begin{proof}
A special case of the cup-cap product in equation \eqref{cup-cap}, reduces when $B_1=\CC = B_2=D$  to
\beq\label{cup-cap4}
\otimes_\CC : \KK^G(A_1, \CC) \otimes  \KK^G(A_2, \CC) \longrightarrow  \KK^G(A_1 \otimes A_2, \CC).
\eeq

 The $G$-comodule action map 
$$a : A_2 \longrightarrow A_1 \otimes A_2.$$
induces a $*$-homomorphism  of $G$-$C^*$-algebras.
Using the fact that K-homology is contravariant with respect to $*$-homomorphisms, we get a canonical abelian group homomorphism
\beq \label{action}
a^* :\KK^G(A_1\otimes A_2, \CC)\longrightarrow  \KK^G( A_2, \CC).
\eeq
Composing $\otimes_\CC$ with $\Delta^*$, we obtain
$$
\circ : \KK^G( A_1, \CC) \otimes  \KK^G(  A_2, \CC) \longrightarrow  \KK^G( A_2, \CC).
$$
The fact that $\circ$ is an action follows from the associativity of the cup-cap product (Theorem 2.14 of \cite{Kas}) and the defining property of the comodule map $\Delta$.

More precisely, the associativity of Kasparov's cup-cap product says that the following diagram commutes,
\begin{equation*} \label{cup-cap-associativity}
\xymatrix @=8pc @ur { \K^G(A_1\otimes A_1) \otimes \K^G(A_2) \ar[d]_{\otimes_\CC} & 
\K^G(A_1)\otimes \K^G(A_1)\otimes \K^G(A_2) \ar[d]_{1\otimes \otimes_\CC} 
\ar[l]^{\otimes_\CC \otimes 1} \\ \K^G(A_1\otimes A_1\otimes A_2) & 
\K^G(A_1)\otimes \K^G(A_1\otimes A_2)
\ar[l]^{\otimes_\CC}}
\end{equation*}
where $\K^G(A)$ denotes $\KK^G(A, \CC)$, $A_1=\CT(G, \eta)$ and $A_2=  \CT(X, f^*(\eta))$.

On the other hand, the defining property of the comodule map $\Delta$ says that the following diagram commutes,
\begin{equation*} \label{cup-cap-associativity}
\xymatrix @=8pc @ur { A_1\otimes A_2 \ar[d]_{1\otimes a} & 
A_2 \ar[d]_{a} 
\ar[l]^{a} \\  A_1\otimes A_1\otimes A_2 & 
A_1\otimes A_2
\ar[l]^{a \otimes 1}}
\end{equation*}

Therefore the induced diagram in equivariant K-homology commutes

\begin{equation*} \label{cup-cap-associativity}
\xymatrix @=8pc @ur { \K^G(A_1\otimes A_2) \ar[d]_{a^*} & 
\K^G(A_1\otimes A_1\otimes A_2) \ar[d]_{1\otimes a^*} 
\ar[l]^{a^* \otimes 1} \\  \K^G(A_2) & 
\K^G(A_1\otimes A_2)
\ar[l]^{a^*}}
\end{equation*}

Therefore one has the commutative diagram,

\beqn
\xymatrix{
\K^G(A_1)\otimes \K^G(A_1)\otimes \K^G(A_2) \ar[r]^{1\otimes(\otimes_\CC)}\ar[d]_{(\otimes_\CC)\otimes 1} & \K^G(A_1)\otimes 
\K^G(A_1\otimes A_2)\ar[r]^{1\otimes a^*}\ar[d]_{\otimes_\CC} & \K^G(A_1)\otimes \K^G(A_2)\ar[d]_{\otimes_\CC} \\
\K^G(A_1\otimes A_1)\otimes \K^G(A_2) \ar[r]^{\otimes_\CC}\ar[d]_{a^*\otimes 1} & \K^G(A_1\otimes A_1\otimes A_2) \ar[r]^{1\otimes  a^*}\ar[d]_{a^*\otimes 1} & \K^G(A_1\otimes A_2)\ar[d]_{a^*}\\
 \K^G(A_1)\otimes \K^G(A_2) \ar[r]^{\otimes_\CC} & \K^G(A_1\otimes A_2) \ar[r]^{a^*} & \K^G(A_2)
}
\eeqn

The top left hand square commutes since the cup-cap product is associative, the bottom right hand square 
commutes since the coaction is coassociative, while the remaining squares commute because 
the cup-cap product is functorial under $G$-$C^*$-homomorphisms.

Therefore the product $\circ$ satisfies the commutative diagram

\begin{equation*} \label{cup-cap-associativity}
\xymatrix @=8pc @ur { \K^G(A_1) \otimes \K^G(A_2) \ar[d]_{\circ} & 
\K^G(A_1)\otimes \K^G(A_1)\otimes \K^G(A_2) \ar[d]_{1\otimes \circ} 
\ar[l]^{\circ \otimes 1} \\ \K^G(A_2) & 
\K^G(A_1)\otimes \K^G(A_2)
\ar[l]^{\circ}}
\end{equation*}
which precisely says that $\circ$ is an action.
\end{proof}

\begin{cor}\label{K-mod2}
Let $G$ be a compact Lie group, and $[\eta]\in\H^3_G(G,\ZZ)$ a primitive cohomology class. Let  $(X, f) \in \eC(G, \eta)$. Then the abelian group $\K^G(X, f^*(\eta)) =\KK^G(\CT(X, f^*(\eta)), \CC)$ admits a  
$\K^G(G,\eta)=\KK^G(\CT(G, \eta), \CC)$-module structure,  induced by the group action map $a:G\times X\map X$. 
\end{cor}

\begin{proof}
Let $(X, f) \in \eC(G, \eta)$. Setting $A_1= \CT(G, \eta)$ and $A_2=  \CT(X, f^*(\eta))$, we see that 
$$
\otimes_\CC: \KK^G( \CT(G, \eta), \CC) \otimes  \KK^G(  \CT(X, f^*(\eta)), \CC) \longrightarrow  \KK^G( \CT(G, \eta)
\otimes \CT(X, f^*(\eta)), \CC).
$$ The group action map $a:G\times X\map X$ is a $G$-map under the adjoint action of $G$, therefore it induces a $*$-homomorphism  of $G$-$C^*$-algebras,
 $$
 a^*:\CT(X, f^*(\eta))\map \CT(G \times X, a^*f^*(\eta)) \cong  \CT(G \times X, p_1^*\eta + p_2^*f^*(\eta))
\cong \CT(G, \eta)\otimes \CT(X, f^*(\eta))
 $$
 since $a^*(f^*\eta) = p_1^*(\eta) + p_2^*f^*(\eta)$ by the primitivity assumption on $\eta$.
 The defining property of the action map $a$ shows that $a^*$ is a comodule map.
The corollary is proved
 by applying Proposition \ref{K-mod}.

\end{proof}

Combining the above with the result of Freed-Hopkins-Teleman \cite{FHT1, FHT2, FHT3}, we have,

\begin{cor}
Let  $(X, f) \in \eC(G, \eta)$, where $G$ is a compact simple Lie group. Then $\K^G(X, f^*(\eta)) $ is a module over the Verlinde algebra $R_\ell(G)$, where 
$\ell$ is the level determined by twist $\eta$. 
\end{cor}

We call such a module over $R_\ell(G)$, a {\em Verlinde module}.

The following is an immediate corollary of Proposition \ref{K-mod}.

\begin{cor}\label{q-mod}
Let $A_1$ be a compact $C^*$-quantum group and 
 $A_2$ be a separable $C^*$-algebra which is also a
comodule for $A_1$.
Then the abelian group $\KK(A_2, \CC)$ is a module for the algebra
$\KK(A_1, \CC)$.
\end{cor}

\subsection{Morphisms of Verlinde modules} 
Here we study morphisms in the category  $\eC(G, \eta)$, and our main result here is
that any morphism in the category  $\eC(G, \eta)$, determines a morphism of Verlinde modules over the Verlinde algebra.

\begin{prop}\label{morphism}
Let $\phi : (X_1, f_1) \longrightarrow  (X_2, f_2) $ be a morphism in $\eC(G, \eta)$, where $G$ is a compact simple Lie group.
Then $\phi_* : \K^G(X_1, f_1^*(\eta))\longrightarrow \K^G(X_2, f_2^*(\eta))$ is a morphism of Verlinde modules.
\end{prop}

\begin{proof}

Let $\phi : (X_1, f_1) \longrightarrow  (X_2, f_2) $ be a morphism in $\eC(G, \eta)$. 
That is, $\phi: X_1 \map X_2$ is an equivariant map such that the following diagram commutes,
\begin{equation*}
\xymatrix{
X_1 \ar[ddr]_{f_1} \ar[rr]^{\phi} && X_2
\ar[ddl]^{f_2} \\ &&\\
& G   }  
\end{equation*}
That is, $f_1=f_2\circ \phi$.

Then it induces a morphism of groups
$$
\phi_* : \K^G(X_1, f_1^*(\eta))\longrightarrow \K^G(X_2, f_2^*(\eta)).
$$
We will show that $\phi_*$ is actually a  morphism of Verlinde modules over the Verlinde algebra. That is, the following diagram commutes,

\beq\label{Kmorphism}
\xymatrix{
\K^G(G, \eta)\otimes \K^G(X_1, f_1^*(\eta))\ar[rr]^{\qquad\circ_1} \ar[dd]_{\id\otimes \phi_*} && 
\K^G(X_1, f_1^*(\eta))\ar[dd]^{\phi_*}\\
\\
\K^G(G, \eta)\otimes \K^G(X_2, f_2^*(\eta))\ar[rr]_{\qquad\circ_2} &&  \K^G(X_2, f_2^*(\eta))
}
\eeq

That is, for $\xi \in \K^G(G, \eta)$ and $x_1 \in \K^G(X_1, f_1^*(\eta))$, the commutativity of the 
diagram above says that 
$$
\phi_*(\xi\circ_1 x_1) = \xi \circ_2 \phi_*(x_1),
$$
where we define $\circ_j$ below.
Now we have the commutative diagram,

\beq\label{action}
\xymatrix{
G\times X_1  \ar[rr]^{\qquad a_1} \ar[dd]_{\id\times \phi} && 
X_1\ar[dd]^{\phi}\\
\\
G\times X_2 \ar[rr]_{\qquad a_2} &&  X_2
}
\eeq 
Therefore $\phi\circ  a_1 = a_2\circ  \phi$. So the induced maps on equivariant K-homology satisfy
$\phi_*\circ  {a_1}_* = {a_2}_*\circ  \phi_*$. As observed earlier, since $\eta$ is primitive, 
$a_j^*(f_j^*(\eta))= p_1^*(\eta) + p_2^*(f^*(\eta))$, $\, j=1,2$. Consider the diagram,

$$
\xymatrix{
\K^G(G, \eta)\otimes \K^G(X_1, f_1^*(\eta))\ar[rr]^{\qquad\otimes_\CC} \ar[dd]_{\id\otimes \phi_*} && 
\K^G(G\times X_1,a_1^* f_1^*(\eta))\ar[dd]_{\id \times \phi_*}\ar[rr]^{{a_1}_*} && \K^G(X_1, f_1^*(\eta))\ar[dd]^{\phi_*}\\
\\
\K^G(G, \eta)\otimes \K^G(X_2, f_2^*(\eta))\ar[rr]_{\qquad\otimes_\CC} &&
\K^G(G \times X_2, a_2^*f_2^*(\eta))\ar[rr]^{\,\,{a_2}_*} &&  \K^G(X_2, f_2^*(\eta))
}
$$
The commutativity of the right square follows from the commutativity of the diagram  in equation \eqref{action}.
The commutativity of the left square follows by naturality of the cup-cap product of Kasparov \cite{Kas, Black}. Therefore we have justified the commutativity of the diagram in equation \eqref{Kmorphism}, proving the proposition.

\end{proof}

\subsection{The May structure on $\eC(G,\eta)$ and a lax May functor}

Let $\Mod(R_\ell(G))$ denote the abelian category of all modules over the Verlinde ring $R_\ell(G)$. We endow $\Mod(R_\ell(G))$ with the symmetric monoidal structure given by the tensor product $\otimes=\otimes_\ZZ$ of abelian groups, as well as the symmetric monoidal structure given by the
direct sum $\oplus$, giving it a May structure.

Recall that a functor $F:(\cC,\otimes_\cC,\mathbb{I}_\cC)\functor(\cD,\otimes_\cD,\mathbb{I}_\cD)$ between strict monoidal categories is called {\em lax monoidal} if there are natural maps $F(C)\otimes_\cD F(D)\map F(C\otimes_\cC D)$ and $\mathbb{I}_\cD\map F(\mathbb{I}_\cC)$ for every pair of objects $C,D\in\cC$ satisfying certain predictable compatibility conditions. 

\begin{thm} The functor $\eF: (\eC(G, \eta), \coprod, \fusion) \longrightarrow ({\rm Mod}(R_\ell(G)),\oplus, \otimes)$, defined by $$\eF(X, f)=\K^G(X, f^*(\eta))$$ is lax monoidal, respecting the May structures on both categories. 
\end{thm} 

\begin{proof}
Let $(X_1, f_1), (X_2, f_2) \in \eC(G, \eta)$ and recall that $\eF(X_i,f_i)=\K^G(\CT(X_i,f_i^*(\eta)))$ for $i=1,2$. Recall that $\K^G(X,f^*(\eta))=\K^G(\CT(X,f^*(\eta)))$. Then  
$$
{\eF}((X_1, f_1) \circledast (X_2, f_2)) = \K^G(X_1\times X_2,  (m\circ (f_1 \times f_2))^*(\eta))
$$
The cup-cap product in equation \eqref{cup-cap4} with  $A_1= \CT(X_1, f_1^*(\eta))$ and $A_2=  \CT(X_2, f_2^*(\eta))$
gives 
\begin{align*}
\otimes_\CC : \K^G(X_1, f_1^*(\eta)) \otimes  \K^G(X_2, f_2^*(\eta))& \longrightarrow \K^G(X_1\times X_2, (f_1\circ p_1)^*(\eta) 
+ (f_2\circ p_2)^*(\eta)) \\
& \cong  \K^G(X_1\times X_2,  (f_1 \times f_2)^* \circ m^*(\eta))\\
& \cong \K^G(X_1\times X_2, (m\circ (f_1 \times f_2))^*(\eta))
\end{align*}
where we have used the primitivity of $\eta$ for the last equality. Therefore we get a natural map
$$
{\eF}(X_1, f_1) \otimes {\eF}(X_2, f_2) \longrightarrow {\eF}((X_1, f_1) \circledast (X_2, f_2)) 
$$ There is also a canonical map relating the unit objects as follows:
$\ZZ {\map}{\fk}(e, \I) = \K^G(e),$ which is the unique unital ring homomorphism. Note that $\ZZ$ is the unit object in the monoidal category $\Mod(R_\ell(G))$ with respect to $\otimes_\ZZ$.

Also
\begin{align*}
{\eF}((X_1, f_1) \coprod (X_2, f_2)) &= \K^G(X_1\coprod X_2, ( f_1 \coprod f_2)^*(\eta))\\
& = \K^G(X_1,  f_1^*(\eta)) \oplus  \K^G(X_2,  f_2^*(\eta))\\
& = {\eF}(X_1, f_1) \oplus  {\eF}(X_2, f_2).
\end{align*}
The distributive property is clear, completing the proof of the theorem.

\end{proof}

\begin{rem}\label{gencat}
If $A$ is a $G$-$C^*$-algebra which is a $G$-coalgebra, define $\eC(A)$ to be the category of all
$G$-$C^*$-algebras that are $G$-comodules over $A$. Then there is a fusion product $\fusion$
on $\eC(A)$, and a lax monoidal functor 
$\cF: (\eC(A), \oplus, \fusion) \longrightarrow ({\rm Mod}(\K^G(A)), \oplus, \otimes)$ defined by $\eF(B)=\K^G(B)$.
The proof is similar to above and will be analysed in future work.
\end{rem}

\section{The category $\eD(G, \eta)$ and quantization}

In this section, we explore when objects in the category $\eC(G, \eta)$ can be quantized.
 To achieve this goal, we
define a closely related category $\eD(G, \eta)$, which objects are triples $(X, E,  f)$ where $(X, f)$ is an
object in the category $\eC(G,\eta)$ and
$E$ is a $G$-equivariant (complex) vector bundle over $X$.
In addition, we assume that $X$ has an {\em equivariant twisted Spinc structure}, that is, the following diagram commutes,
\beq
\xymatrix{
X_G\ar[r]^\nu \ar[d]_{f_G} & BSO\ar[d]^{\pi}\\
G_G\ar[r]_{\eta_G} & K(\ZZ, 3).
}
\eeq 
Here $K(\ZZ, 3)$ is the 3rd Eilenberg-Maclane space, $\nu$ is a continuous map classifying the stable normal bundle of 
the Borel construction $X_G=  EG \times_G X$ or equivalently, classifying the equivariant stable
normal bundle of $X$. Similarly $G_G =  EG\times_G G$ where $G$ acts on itself by conjugation
and $f_G$ is the map induced by $f$. Moreover $\pi$ is a continuous map determined by the Stieffel-Whitney class, up to homotopy.
Such a choice is fixed.
This implies that 
\beq\label{FW}
f_G^*(\eta_G) + W_3(X_G) = 0,
\eeq
where $\eta_G$ is the induced twisting on $G_G$.
This is the analogue of the Freed-Witten anomaly cancellation condition for D-branes in type II superstring theory, \cite{FW}. Geometrically, equation \eqref{FW} means that there is an equivariant isomorphism,
\beq
f^*(\eK_\eta) \cong {\rm Cliff}(TX)\otimes \eK,
\eeq
where $\eK_\eta$ is the algebra bundle of compact operators on $G$ determined by $\eta$, $ {\rm Cliff}(TX)$ denotes the Clifford algebra bundle associated to the tangent bundle of $X$, and $\eK$ denotes
the algebra of compact operators on a $G$-Hilbert space.
The objects of $\eD(G, \eta)$ are an equivariant analogue of twisted geometric $K$-cycles in \cite{Baum, Wang}.
The key observation made here is that in this special case, this category has a richer structure than usual,
given by the May structure.

$(G, {\bf 1}, \id:G\map G)$ is a final object in the category $\eD(G, \eta)$, where $\bf 1$ is the trivial line bundle over $G$. 
The morphisms of $\eD(G, \eta)$ are explicitly described in the text. In particular, a compact quasi-Hamiltonian $G$-manifold $(M, \omega, \Phi)$ determines the object $(M, {\bf 1}, \Phi)$ in $\eD(G, \eta)$, by a result in \cite{AM}. Clearly $\eD(G, \eta)$ is much larger, and it 
is closed under disjoint union $\coprod$, a dual operation, the fusion product $\fusion$ and 
also $G$-vector bundle modification, all of which will be explained in the text. Here we mention 
that the fusion product is $(X_1, E_1, f_1) \fusion (X_2, E_2, f_2) = (X_1 \times X_2, E_1\boxtimes E_2, m\circ (f_1 \times f_2))$, for objects $(X_j, E_j, f_j) \, j=1, 2$ in $\eD(G, \eta)$.
We verify that $\eD(G, \eta)$ is a strict monoidal category
that has a {May structure} given by $\coprod$ and $\fusion$.

 Let $(X_1, E_1, f_1) $ and $(X_2, E_2, f_2)$ 
denote equivariant twisted geometric K-cycles in  $\eD(G, \eta)$. They are said to be {\em isomorphic}
if there is an equivariant diffeomeophism
$\phi: X_1 \map X_2$ such that the following diagram commutes,
\begin{equation*}
\xymatrix{
X_1 \ar[ddr]_{f_1} \ar[rr]^{\phi} && X_2
\ar[ddl]^{f_2} \\ &&\\
& G   }  
\end{equation*}
That is, $f_1=f_2\circ \phi$. Moreover, it is assumed that there is an equivariant isomorphism
$\phi^*(E_2)\cong E_1$.

We now impose an equivalence relation $\sim$ on $\eD(G, \eta)$, generated by isomorphism and  
the following three elementary relations:
\begin{enumerate}

\item  {\bf Direct sum -  disjoint union.}
 Let $(X, E_1, f) $ and $ (X, E_2, f)$ 
denote equivariant twisted geometric K-cycles in  $\eD(G, \eta)$ with the same equivariant
twisted Spinc structure,
then their disjoint union is the equivariant twisted geometric K-cycle given by the direct sum,
\[
(X, E_1, f) \coprod  (X, E_2, f) \sim (X, E_1\oplus E_2, f).
\]

\item  {\bf Equivariant bordism.}
Given two equivariant twisted geometric K-cycles $(X_1, E_1, f_1) $ and $(X_2, E_2, f_2)$ such that 
there exists an equivariant twisted Spinc manifold with boundary  $W$, an equivariant vector bundle
$E$ over $W$ and a $G$-map $f: W \map G$ such that 
\[
\partial W=  -X_1 \coprod   X_2, \qquad \partial E= E_1 \coprod E_2, \qquad f\big|_{\partial W} = 
f_1\coprod f_2.
\]
Here $-X$ denotes  the $G$-manifold  $X$  with the  opposite equivariant twisted Spinc structure.
Then $(W, E, f)$ is said to be an equivariant bordism between the equivariant twisted geometric K-cycles
$(X_1, E_1, f_1) $ and $(X_2, E_2, f_2)$.\\

\item   {\bf Equivariant Spinc vector bundle modification.}
 Let 
  $(X, E, f)$ be an equivariant twisted geometric K-cycle and $V$ an equivariant 
  a  Spinc vector bundle over $X$ with 
  even dimensional fibers.  Denote by $\underline{\RR}$ the trivial rank one real 
  vector bundle. Choose an invariant Riemannian metric on $V\oplus \underline{\RR}$, let
  $$\widehat{X}= S(V\oplus \underline{\RR})$$  be the total space of 
  the sphere bundle of $V\oplus \underline{\RR}$, which is a $G$-manifold. 
  Then
   the vertical tangent bundle $T^{vert}(\widehat{X})$ of $ \widehat{X}$ admits a natural 
   equivariant Spinc structure
   with an associated $\ZZ_2$-graded equivariant spinor bundle  $S^+_V\oplus S^-_V$ . Denote by
  $\pi: \widehat{X} \to X$   the projection,  which is equivariantly  K-oriented.  Then
  the equivariant Spinc vector bundle modification of $(X, E, f)$ along the equivariant 
  Spinc vector bundle $V$, is the equivariant twisted geometric K-cycle
  $( \widehat{X}, \pi^*E\otimes S^+_V, f\circ \pi)$.
\end{enumerate}

\begin{defn} \label{twisted:geo} Denote by  $\K^G_{geo, \bullet}(G, \eta) = \eD(G, \eta)/\sim$ the 
geometric equivariant twisted K-homology. Addition in $\K^G_{geo, \bullet}(G, \eta) $ is given by disjoint union - 
direct sum relation. Note that the equivalence relation $\sim$ preserves the parity
of the dimension of the underlying equivariant twisted Spinc manifold. Let 
$\K^G_{geo}(G, \eta) $ denote the subgroup of $\K^G_{geo, \bullet}(G, \eta)$
determined by all geometric cycles with even dimensional
equivariant twisted Spinc manifolds. 
\end{defn}

Define the fusion product $\fusion$ of  equivariant twisted geometric K-cycles 
$(X_1, E_1, f_1) $ and $(X_2, E_2, f_2)$ as
\[
(X_1, E_1, f_1) \fusion (X_2, E_2, f_2)= (X_1 \times X_2, E_1\boxtimes E_2, m\circ (f_1\times f_2)).
\]
Here $m: G \times G \map G$ denotes the multiplication on the group, which is an equivariant map
with respect to the conjugation action of $G$ on itself. Then $(\eD(G, \eta), \coprod, \fusion)$ is 
a May category and we have,

\begin{prop}
The geometric equivariant twisted K-homology group $\K^G_{geo}(G, \eta) $ is a ring, with product 
induced by the fusion product $\fusion$.
\end{prop}

There is a quantization functor $\cQ: \eD(G, \eta) \map \K^G(G, \eta)$ which we recall here.
Recall that every equivariant twisted geometric K-cycle $(X, E, f)$ has a fundamental class 
$[X] \in \K^G(X, {\rm Cliff}(TX))$ as defined in \cite{Kas}, which is a Dirac type operator. Then 
$\cQ(X, E, f) = f_*([E]\cap [X]) \in K^G(G, \eta)$. Then we have

\begin{thm} The quantization  functor $\cQ: \eD(G, \eta) \longrightarrow \K^G(G, \eta)$, defined by
 $$\cQ(X, E, f)=  f_*([E]\cap [X]) $$ is monoidal, respecting the May structures on both categories. 
\end{thm} 

\begin{proof}
Let $(X_1,E_1,  f_1), (X_2, E_2, f_2) \in \eD(G, \eta)$ and recall that 
$\cQ(X_i,E_i, f_i)={f_i}_*([E_i]\cap [X_i])$ for $i=1,2$. Then  
\begin{align*}
{\cQ}((X_1, E_1, f_1) \circledast (X_2, E_2, f_2)) &=  (m\circ (f_1 \times f_2))_*([E_1\boxtimes E_2] \cap [X_1\times X_2])\\
&= m_*({f_1}_*([E_1] \cap [X_1]) \times {f_2}_*([E_2] \cap [X_2]))\\
& = {f_1}_*([E_1] \cap [X_1]) \circ {f_2}_*([E_2] \cap [X_2])\\
& = {\cQ}(X_1, E_1, f_1)\circ {\cQ}(X_2, E_2, f_2)
\end{align*}
There is also a canonical map relating the unit objects as follows:
$\cQ(e, {\bf 1}, \I) = \I_*({\bf 1}\cap[e]) =1$ which is the unique unital ring homomorphism. 
Also by the disjoint union-direct sum property, one has
\begin{align*}
{\cQ}((X_1, E_1, f_1) \coprod (X_2, E_2, f_2)) &= {\cQ}((X_1\coprod X_2, E_1\coprod E_2, f_1 \coprod f_2))\\
& = ( f_1 \coprod f_2)_*([E_1\coprod E_2] \cap [X_1\coprod X_2])\\
& = ( f_1)_*([E_1] \cap [X_1]) \coprod   (f_2)_*([E_2] \cap [X_2])\\
& = {\cQ}(X_1, E_1, f_1) +  {\cQ}(X_2, E_2, f_2),
\end{align*}
since $ [X_1\coprod X_2] =  [X_1]\coprod [X_2]$ and $[E_1\coprod E_2] = [E_1]\coprod [E_2]$.
The distributive property is clear, completing the proof of the theorem.

\end{proof}

Finally, the following is a special case of a more general result,

\begin{prop}
The quantization  functor $\cQ$ induces an isomorphism of equivariant K-homology 
rings, $$\K^G_{geo}(G, \eta) \cong \K^G(G, \eta).$$
\end{prop}

This is a special case of a more general theorem which will be proved elsewhere, which uses a hybrid of 
techniques in  \cite{Baum, Wang} and \cite{Baum2}. 

An impact of this 
result is that the May structure on the category $\eD(G, \eta)$ induces the algebra structure on $\K^G(G, \eta)$,
which by \cite{FHT1, FHT2, FHT3} is just the Verlinde algebra $R_\ell(G)$.\\

\section{Outlook and questions}

For $G$ as in the paper, recall that the Verlinde algebra $R_\ell(G)$
consists of equivalence classes of positive energy representations (at a fixed level $\ell$)
of the free loop group $LG$, \cite{PS}. The theorem of Freed-Hopkins-Teleman \cite{FHT1, FHT2, FHT3}
establishes an explicit isomorphism between $\K^G(G,\eta)$ and $R_\ell(G)$, where $\eta$ is a degree 3-twist
on $G$ that determines the level $\ell$. This isomorphism is given by a projective family of 
second quantized supersymmetric Dirac operators, coupled to a positive energy representation. 

Given $(X, f)$ an object in $\eC(G, \eta)$, is it possible to construct a representation theoretic 
group $R_\ell(X, f)$, which is a module over $R_\ell(G)$, together with an explicit isomorphism 
between $K^G(X, f^*(\eta))$ and $R_\ell(X, f)$ that generalizes the Freed-Hopkins-Teleman 
isomorphism?

The algebraic properties of the Verlinde algebra $R_\ell(G)$ are well understood by now. 
Given $(X, f)$ an object in $\eC(G, \eta)$, it would be interesting to understand the algebraic 
properties of the Verlinde module $K^G(X, f^*(\eta))$.

\appendix

\newpage

\section{More examples of objects in the category $\eC(G, \eta)$}

 \begin{ex} (1st Iterated product) 
 
 Consider $G^{r} = \Hom(F^{r},G)$, where $F^{r}$ is the free group on $r$ generators, with $G$-action
 given by the diagonal $G$-action. Consider the smooth map 
 \begin{align*}
 \lambda & : G^{r} \longrightarrow G\\
 & (g_1, \ldots, g_r)  \longrightarrow \prod_{j=1}^n g_j
 \end{align*}
which is equivariant for the adjoint action of $G$ on $G$. Therefore $(G^{r},  \lambda) \in \eC(G,\eta)$
for any $G$-vector bundle over $G^r$.

Because $\eta$ is a primitive class, by induction on $r$ one sees that 
$$
\lambda^*(\eta) = p_1^*(\eta) + \cdots + p_r^*(\eta).
$$
where $p_j$ denotes the projection to the $j$-th factor of $ G^{r}$ for $j=1, \ldots , r$.
 \end{ex}

  \begin{ex}  (2nd Iterated product) 
  
  In the notation above, consider $G^{2r} = \Hom(F^{2r},G)$, with $G$-action
 given by the diagonal $G$-action. Consider the smooth map
 \begin{align*}
 \tau & : G^{2r} \longrightarrow G\\
 & (g_1, h_1 \ldots, g_r, h_r)  \longrightarrow \prod_{j=1}^n [g_j, h_j]
 \end{align*}
 where $ [g_j, h_j] = g_j h_j g_j^{-1} h_j^{-1}$ denotes the group commutator.
Then $\tau$ is an equivariant map for the adjoint action of $G$ on $G$. 
Therefore $(G^{2r},  \tau) \in \eC(G,\eta)$. 
Also one computes that
$$
\tau^*(\eta) =  0.
$$

 \end{ex}

\begin{ex} (Products of spaces in $\eC(G,\eta)$ part 1) 

Suppose that $(X_j,  f_j)  \in \eC(G,\eta)$
for $j=1, \ldots ,r$. Then the Cartesian product 
$$
\prod_{j=1}^r f_j \colon \prod_{j=1}^r X_j \longrightarrow G^r.
$$
Composing with the map $\lambda$ above, we see that 
$$(\prod_{j=1}^r X_j, \,\lambda\circ (\prod_{j=1}^r f_j ))
\in \eC(G,\eta)$$ and  $$(\prod_{j=1}^r X_j,  \,\lambda\circ (\prod_{j=1}^r f_j ))
\in \eC(G,\eta).$$

Finally, using the 1st iterated product example, we see that
$$
 (\prod_{j=1}^r f_j )^*\lambda^*(\eta) = f_1^*p_1^*(\eta) + \ldots +  f_r^*p_r^*(\eta).
$$

\end{ex}

 \begin{ex} (Products of spaces in $\eC(G,\eta)$ part 2) 
 
 Suppose that $(X_j, f_j)$ are objects in  $\eC(G,\eta)$
for $j=1, \ldots ,2r$. Then the Cartesian product 
$$
\prod_{j=1}^{2r} f_j \colon \prod_{j=1}^{2r} X_j \longrightarrow G^{2r}.
$$
Composing with the map $\tau$ above, we see that 
$$(\prod_{j=1}^r X_j, \,\tau\circ (\prod_{j=1}^r f_j ))
\in \eC(G,\eta)$$ and  $$(\prod_{j=1}^r X_j,  \,\tau\circ (\prod_{j=1}^r f_j ))
\in \eC(G,\eta).$$

Using the 2nd iterated product example$$
 (\prod_{j=1}^{2r} f_j )^*\tau^*(\eta) =0.
$$

\end{ex}

\begin{ex}(Examples from Lie groups)

Let $G$ be a compact Lie group, and $T$ a torus subgroup of $G$, for example 
the maximal torus. Then there is a canonical map,
\begin{align*}
p &: G/T\times T\map G\\
&  (gT,t) \mapsto gtg^{-1}
\end{align*}
The smooth map $p$ is equivariant for the action of $G$ on $G/T\times T$ given by,
\begin{align*}
G \times (G/T\times T) &\longrightarrow G/T \times T\\
( g_1, (gT, t)) &\longrightarrow (g_1gT, t)
\end{align*}
and for the adjoint action of $G$ on $G$. Therefore $(G/T \times T,  p) \in \eC(G,\eta)$.

The principal torus bundle 
$T \to G \to G/T$ has first Chern class $a \in H^2(G/T; \widehat{T})$, where $\widehat T$ denotes
the Pontryagin dual of $T$. Also, the 
universal cover of $T$ determines a class $b \in H^1(T; \widehat T)$. Then a calculation shows that 
$$
p^*(\eta) = p_1^*(a) \cup p_2^*(b) \in H^3(G/T \times T; {\mathbb Z})
$$
where we have used the pairing $\widehat{T} \times \widehat{T} \to {\mathbb Z}$ given by the dot product,
$(\underline{n}, \underline{m}) \to \underline{n}.\underline{m}$. 

\end{ex}

%---------------bibliography-------------------------------------------

\end{document}